\documentclass[12pt,a4paper,widepage, reqno, oneside]{amsart}
\usepackage{amssymb}
\usepackage{enumerate}
\usepackage[toc,page]{appendix}
\usepackage{fancyhdr}

\usepackage[dvipsone]{color}
\usepackage[dvips]{graphicx}

\usepackage{stmaryrd}
\usepackage[hypertex,linkcolor=blue,colorlinks,hypertexnames=false,plainpages=false]{hyperref}

\newtheorem{thm}{\noindent\bf Theorem}[section]
\newtheorem{prop}[thm]{\noindent \bf Proposition}
\newtheorem{defin}[thm]{\noindent \bf Definition}
\newtheorem{lemma}[thm]{\noindent \bf Lemma}
\newtheorem{corol}[thm]{\noindent \bf Corollary}

\definecolor{blue2}{rgb}{0,0,0.6}
\definecolor{turck}{rgb}{0,0.5,0.5}
\definecolor{sturck}{rgb}{0,0.9,0.8} 
\definecolor{orange}{rgb}{1,0.4,0}
\definecolor{bord}{rgb}{1,0.3,0.3}
\definecolor{violet}{rgb}{0.4,0.2,0.8}
\definecolor{sblue}{rgb}{0.2,0.4,1.}
\definecolor{yellowa}{rgb}{1.,0.8,0}
\definecolor{greena}{rgb}{0.2,0.8,0}
\definecolor{greenb}{rgb}{0,0.5,0.2}
\definecolor{greenc}{rgb}{0.1,0.7,0.3}
\definecolor{greend}{cmyk}{0.7,0,0.3,0}
\definecolor{greene}{cmyk}{0.7,0,0.3,0.1}
\definecolor{greya}{cmyk}{0,0,0,0.1}
\definecolor{greyb}{cmyk}{0,0,0,0.2}
\definecolor{greyc}{cmyk}{0,0,0,0.3}
\definecolor{greyd}{cmyk}{0,0,0,0.4}
\definecolor{greye}{cmyk}{0,0,0,0.5}
\definecolor{greyf}{cmyk}{0,0,0,0.6}
\definecolor{greyg}{cmyk}{0,0,0,0.7}
\definecolor{greyh}{cmyk}{0,0,0,0.8}
\definecolor{greyi}{cmyk}{0,0,0,0.9}

\newcommand{\linetwo}[2]{{\substack{#1 \\ #2}}} 
\newcommand{\sumtwo}[2]{\sum_{\substack{#1 \\ #2}}} 

\newcommand{\nn}{\nonumber}

\newcommand{\C}{\varpi}

\newcommand{\dis}{\displaystyle}
\newcommand{\al}{\alpha}
\newcommand{\eps}{\epsilon}

\newcommand{\La}{\Lambda}
\newcommand{\la}{\lambda}
\newcommand{\ga}{\gamma}
\newcommand{\Ga}{\Gamma}

\newcommand{\si}{\sigma}

\newcommand{\Zz}{\mathbb{Z}}

\newcommand{\Rr}{\mathbb{R}}

\newcommand{\Qq}{\mathbb{Q}}

\newcommand{\cA}{\mathcal{A}}

\newcommand{\cB}{\mathcal{B}}

\newcommand{\cD}{\mathcal{D}}
\newcommand{\cE}{\mathcal{E}}

\newcommand{\cG}{\mathcal{G}}
\newcommand{\cH}{\mathcal{H}}

\newcommand{\cM}{\mathcal{M}}
\newcommand{\cN}{\mathcal{N}}

\newcommand{\cR}{\mathcal{R}}
\newcommand{\cS}{\mathcal{S}}

\newcommand{\Ii}{\text{\bf 1}}

\newcommand{\e}{\mathrm{e}} 
\newcommand{\sign}{\mathrm{sign}}
\newcommand{\dist}{\mathrm{dist}}

\newcommand{\mmmintone}[1]{{\dis{\int\kern -.43cm
-}}_{\kern-.21cm\substack{#1}}\;\;}
\newcommand{\mmmintwo}[2]{{\dis{\int\kern -.43cm
-}}_{\kern-.21cm\substack{#1}}^{\substack{#2}}\;\;}
\newcommand{\submint}{{\scriptstyle{\int\kern -.66em -}}}
\newcommand{\submintone}[1]{{\scriptstyle{\int\kern -.66em
-}}_{\tiny{\kern-.21em\linetwo{}{\substack{#1}}}}}
\newcommand{\fracmint}{{\textstyle{\int\kern -.88em -}}}
\newcommand{\fracmintone}[1]{{\textstyle{\int\kern -.88em
-}}_{\tiny{\kern-.34em\substack{#1}}}\;}

\newcommand{\ba}{\begin{array}}
\newcommand{\ea}{\end{array}}
\newcommand{\bq}{\begin{eqnarray}}
\newcommand{\eq}{\end{eqnarray}}
\newcommand{\bqw}{\begin{eqnarray*}}
\newcommand{\eqw}{\end{eqnarray*}}

\newcommand{\und}{\underline}
\newcommand{\thvd}[2]{}

\newcommand{\tivd}[2]{}

\long\def\notes#1{\ifinner
         {\tiny #1}
         \else
          \marginpar{\protect\tiny #1}%
          \fi}%
\title[]{Study of a long range perturbation of a one-dimensional Kac model}

\author{Marzio Cassandro}
\address{Universit\`a di Roma ``La Sapienza'',
P.le A. Moro, 00185 Roma, Italy.} \email{marzio.cassandro@roma1.infn.it}

\author{Immacolata Merola}
\address{Dipartimento di Matematica Pura ed Applicata, Universit\`a di L'Aquila, 67100
L'Aquila, Italy.} \email{immacolata.merola@dm.univaq.it}

\author{Maria Eulalia Vares}
\address{Centro Brasileiro de Pesquisas F\'\i sicas (CBPF), Rua Dr. Xavier Sigaud 150, 22290-180 Rio de Janeiro, RJ, Brasil. }\email{eulalia@cbpf.br}

\begin{document}

\maketitle
\setcounter{equation}{0}
\setcounter{section}{0}

\setcounter{page}{1}
\setcounter{part}{0}


\date{}
\begin{abstract}
\noindent We consider a one dimensional ferromagnetic Ising spin system with interactions that correspond to a $1/r^2$ long range perturbation of the
usual Kac model. We apply a coarse graining procedure widely used for higher-dimensional finite range Kac
potentials to describe the basic properties of the system and the relation with the mean field theory.
\smallskip

\end{abstract}

\vskip .5 cm

\section{Introduction}

\noindent We consider a one dimensional Ising spin system on $\mathbb{Z}$ interacting by a long range perturbation of the
usual Kac model. More precisely, for a small positive parameter $\gamma$, the coupling $J(r)$ between spins at distance $r$
is given by $\gamma$ if $|r| \le (2\gamma)^{-1}$, and by $\la/r^{2}$ otherwise, where $\la>0$ is fixed. Applying the perturbative
scheme around the mean field developed in \cite{CP} for finite range Kac potentials in dimensions $d\ge 2$ (see also \cite{e-book})
and following the notion of contours introduced by Fr\"{o}hlich and Spencer in \cite{FS} as implemented in \cite{CFMP}, we study basic
properties of this model for small but finite $\gamma$.

The main properties of percolation and phase transitions for one-dimensional ferromagnetic models with long range interactions
have been established in the seminal papers \cite {FS}, \cite {NS}, \cite{AN}. Particularly relevant are the results obtained in
\cite{FS}, \cite{ACCN}, and \cite{IN}, where the role of long and short range components of the interactions
has been singled out. When $r^2J(r) \to \lambda \in (0,\infty)$ as $r \to
\infty$, they establish the existence of phase transition, prove the
discontinuity of the magnetization at the critical point $\beta_c$ (the so called Thouless effect)
and (among other things) determine the limiting value of $\beta_c$ when a short range interaction, say $J(1)$,
tends to infinity. For some related more recent results see e.g. \cite{M,MSV}.

Our model belongs to this class, except that the strong short range interaction is replaced by the standard finite Kac potential that
in the limit $\gamma \to 0$ gives the mean field model with spontaneous magnetization
\begin{equation*}
m_\beta=\tanh \beta m_\beta
\end{equation*}
for $\beta >1$.

The existence of a finite bound for $\beta_c(\gamma,\lambda)$, uniform in $\gamma$ for $\gamma$ small follows by putting together the results of \cite{NS} for
one-dimensional independent site-bond percolation and the inequalities between percolation and Ising models, obtained in \cite{ACCN} through the Fortuin-Kasteleyn representation (\cite{FK}). Instead, we will exploit a coarse graining procedure (widely used to study finite range Kac systems) to get not only a direct proof of this bound, but also a detailed description of the typical configurations. This approach allows to display the relations with the mean field theory. We show the existence of $\bar\beta(\lambda)$ so that for all $\gamma$ sufficiently small:

\noindent a) $\beta_c(\gamma,\lambda) < \bar\beta(\lambda)$.

\noindent b) For all $\beta> \bar{\beta}(\lambda)$,
the magnetization under the extremal Gibbs measure with $+1$ ($-1$) external
condition is close to the mean field value $m_\beta$ ($-m_\beta$ resp.).

\noindent c) For all $\beta>\bar{\beta}(\lambda)$ and
external conditions cf. Definition \ref{def:cS}, or still as in b) above, the typical configurations
exhibit large intervals (of length $\ge\exp{(\frac {c(\beta,\lambda)}{\gamma\ln{1/\gamma}})}$) with
magnetization close to $+m_\beta$ or $-m_\beta$ interrupted by fluctuations of the opposite
sign of order $o(\gamma^{-1})$.

We believe that with a proper implementation of the multiscale approach introduced by \cite{IN}, the upper bound $\bar\beta(\lambda)$ might
be improved and that the spontaneous magnetization should stay close to the mean field value for any $\beta >\beta_c(\ga,\la)$. Nevertheless,
it is not clear if the method can be suitably applied to our case. Dealing with coarse grained configurations imposes difficulties in the
treatment of Peierls type estimates, and the contour methods implemented here do not provide an optimal result.

\vskip 0.5cm

\section{The model}

\noindent We consider a spin system on $\Zz$: $i\in \Zz$
 $\si(i)\in \cS_1:=\{-1,+1\}$.
 Given $\bar\si \in \cS_1^{\Zz}$ and  $\La\subset \Zz$ finite,
the model on $\cS_\La:=\cS_1^{\La}$ is defined by the Hamiltonian

\begin{eqnarray}
  \label{def:H}
H_\La(\si_\La|\bar\si)= -\frac{1}{2} \sum_{i,j\in \La\cap \Zz} J_\ga(|i-j|)\si(i)\si(j) -
\sumtwo{i\in \La\cap \Zz}{j\in \La^c\cap \Zz} J_\ga(|i-j|)\si(i)\bar\si(j),
  \end{eqnarray}
  with
  \begin{eqnarray}
  J_\ga(|i-j|)&:=& {\ga} \Ii_{[|i-j|\le \ga^{-1}/2]}+
\la \frac{\Ii_{[|i-j|> \ga^{-1}/2]}}{|i-j|^{2}}\nn
\\
&\equiv & \ga (J^{(0)}+\tilde\la J^{(1)})(\ga(i-j)),
  \label{def:J}
\end{eqnarray}
where $\tilde\la= \la \ga$,
\begin{equation}
\label{def:JJ}
J^{(0)}(r)=\Ii_{[|r|\le1/2]} \;\;\; \text{and}\;\;\; J^{(1)}(r)=\frac{1}{r^2} \Ii_{[|r|>1/2]},
\end{equation}
and $\Ii_A$ stands for the indicator function of the set $A$.
 The parameter $(2\ga)^{-1}$, which gives (in microscopic scale) the range of the basic short range mean field
 interaction, is a positive even integer assumed to be large throughout. The Gibbs measure at inverse temperature $\beta$ on the finite volume $\La$ and with external condition $\bar\si$ is given by
\begin{eqnarray*}
\mu_{\La,\beta,\gamma}(\si_\La|\bar \si)=\frac{\dis{e^{-\beta H_\La(\si_\La|\bar\si)}}}{Z_\La(\bar \si)},
\end{eqnarray*}
where
\begin{eqnarray*}
Z_{\La,\beta,\gamma}(\bar \si):= \sum_{\si_\La \in \mathcal{S}_\La} e^{-\beta H_\La(\si_\La|\bar \si)}.
\end{eqnarray*}
To avoid heavy notation we usually omit the parameter $\lambda$ that
appears in (\ref{def:J}). Sometimes, whenever no confusion is added, we also omit
$\gamma$ or the inverse temperature $\beta$ from the notation.

In this paper, we shall always work in the so-called mean-field phase transition region for  $J^{(0)}$, i.e., we assume $\beta >1$ throughout. Let $m_\beta$  denote the mean field value of the magnetization at temperature $1/\beta$, i.e., the positive solution of the mean field equation:
\begin{eqnarray}
\label{eq:mbeta}
m_\beta =\tanh (\beta m_\beta).
\end{eqnarray}
Given $b>0$ we write $\tilde\beta(b)$ for the unique value of $\beta$ that solves
\begin{equation}
\label{eq:betalambda}
\beta m^2_{\beta}=b,
\end{equation}
with $m_\beta$ the non-zero solution of \eqref{eq:mbeta},

In this context, we have:
\begin{thm}
\label{thm:main}
There exists $\bar b$ (independent of $\lambda, \gamma$) so that the following holds:

\noindent a)  For any $\lambda>0$ we can find $\gamma_0(\lambda)>0$ so that for any $\gamma <\gamma_0(\lambda)$, the system with parameters $\gamma,\lambda$ exhibits phase transition and the critical inverse temperature
 $\beta_c(\ga,\la)$ satisfies
\begin{eqnarray}
\label{eq:m1}
\beta_c(\ga,\la)\le \tilde\beta(\bar b/\lambda)=:\bar\beta(\la).
\end{eqnarray}

\noindent b) If $\beta >\bar\beta(\lambda)$, then
\begin{eqnarray}
\label{eq:m3}
{\lim_{\ga \to 0}\mu^{{+}}_{\beta,\gamma}(\si_0)=m_\beta}
\end{eqnarray}
where $\mu^{+}_{\beta,\gamma}:=\lim_{\La\to \mathbb{Z}}\mu_{\La,\beta,\gamma}(\cdot|+\und 1)$ with $+\und 1 $ denoting
the configurations $\bar \si_j=+1$ for all $j$ and, as usual, $\mu(f)$ denotes the integral of $f$ with respect to $\mu$.
\end{thm}

\noindent {\bf Remark.} The notion of criticality is the standard one, marking the transition from uniqueness to multiple infinite volume Gibbs measures. The classical Dobrushin uniqueness condition (see \cite{D70}) tells that $(1+4\la\gamma)^{-1} <\beta_c(\ga,\la)$.
\vskip 0.5 cm
\noindent Taking into account the result in \cite{ACCN} where it is shown that
if $\lim_{r\to \infty}r^2 J(r)= \la$ exists, with $0<\lambda<\infty$, then
the following dichotomy holds\footnote{{\bf Notational Remark.} For a given interaction $J(\cdot)$, the Hamiltonian in (\ref{def:H}) corresponds to twice that in \cite{ACCN,IN}.}: $\mu^{+}_{\beta,\gamma}(\si_0)=0$ or
$\mu^{+}_{\beta,\gamma}(\si_0)\ge (2\beta \la)^{-1/2}$, we have at once that $\bar b\ge 1/2$. It would be very interesting to extend the analysis to all values of $\beta$ larger than $\tilde\beta(1/(2\lambda))$,  but our techniques do not allow this for the moment. (Our proof works for any $\bar b > 7$.) We should also notice that in \cite{ACCN}, see also \cite{AN} and \cite{IN}, the more general context of Potts models is considered.

The plan of the paper is as follows: in section \ref{sec:coarsegraining} we exploit a coarse graining procedure
widely used in the study of Kac systems (see \cite{e-book}) to describe the configurations in terms of $\{-1,0,1\}$-valued spin
variables, and state our main theorem in this context. In section \ref{sec:proofthm1} we extend to our case the notion of contours introduced in \cite{FS}, but we follow the implementation given in \cite{CFMP}, that is better suited to control the contributions of the zero components of these new  spins. In section \ref{sec:proofof m31} we prove the upper bound for $\beta_c(\ga,\la)$ via a
Peierls argument. In section \ref{sec:basic-estimates} we prove the {free-energy} estimates necessary to implement the Peierls argument.

\vskip 0.5cm
\section{Coarse graining}
\label{sec:coarsegraining}

\noindent In the sequel we will introduce three new scales, $\ell_0<\ell_-<\ell_+$, where $\ell_0$, $\ell_-/\ell_0$ and $\ell_+/\ell_-$
are positive integers, all tending to $\infty$ as $\ga \to 0$. We also assume  $\ell_0, \ell_-, \ell_+ \in \ga^{-1}\Qq$:
\begin{equation}
\label{def:ell0}
\ell_0:= \delta_0\ga^{-1} <<\ell_-:= \delta_-\ga^{-1}<< (2\ga)^{-1} <<\ell_+:= \delta_+\ga^{-1}.
\end{equation}
For our proof to work  $\delta_0,\delta_-,\delta_+$ should satisfy some relations; an example is given in \eqref{eq:deltachoice}.

\vskip .5cm \noindent

\noindent {\bf Notation.} For $B\subset \Zz$ finite, $|B|$ denotes its cardinality. For any $x\in \ell_*\Zz$, where $*$ stands for $0,-$ or $+$,
we write $C^{*}_x= [x, x+\ell_*)\cap \mathbb{Z}$, also called $\ell_*$-blocks in the sequel.\footnote{$\ell_*$-block may also refer to any interval that is measurable with respect to the partition generated by the $\{C^{*}_x\}$, also called $\ell^*$-measurable interval.} We also set
\begin{eqnarray}
\label{def:m-si}
m^{\ell_*}(x;\si)=
\frac{1}{{| C^*_x|}}\sum_{i\in C^*_x} \si(i), \hskip1cm   \si\in \cS,
\end{eqnarray}
where now $*$ stands for 0 or $-$.  Thus $m^{\ell_*}(x;\si)$ takes  values in $\{-1,-1+\frac{2}{\ell_*},\dots, 1\}=:\cM_*$.
We call $\cM_{*,\La}:=\cM_*^{|\La\cap \ell_*\Zz|}$.
\vskip .5cm \noindent

For any configuration $\sigma \in \mathcal {S}_\Lambda$ we now define the coarse-grained variables $\eta_\Lambda=(\eta^\psi(x,\sigma)\colon C_x^+\subset \La)$. The variable $\eta^\psi(x,\sigma)$ provides information on how close the averages $m^-(\cdot,\sigma)$ are to the non-zero solutions of the mean field equation $\pm m_\beta$, over the $\ell_+$-block $C^+_x$. They depend also on a parameter $\psi$ related to the accuracy,
and which will be thought as suitably small in comparison with $m_\beta$.

\begin{eqnarray}
\label{def:eta}
\eta^\psi(x, \si)= \left\{
                \begin{array}{ll}
                  -1, & \hbox{if } \;\; \sup_{y\in C^+_{x}\cap \ell_-\Zz}|m^{\ell^-}(y;\si)+m_\beta|<\psi,
 \\
 \\
                  +1, & \hbox{if } \;\; \sup_{y\in C^+_{x}\cap \ell_-\Zz}|m^{\ell^-}(y;\si)-m_\beta|<\psi,
      \\
\\
                  0, & \hbox{otherwise,}
                \end{array}
              \right.
\end{eqnarray}
where we take $\psi=\frac{1}{N} (m_\beta)^2$ for $N$ large.

\vskip 0.5 cm
\noindent {\bf Remark.} The parameter $\psi$ is fixed in the proof below, and therefore sometimes omitted in the notation. On the other
hand, a careful examination of the estimates shows that one can indeed make $\psi=\psi(\gamma)$ tend
to zero with $\gamma$.

\begin{defin}
\label{def:cS}
We set
       \begin{eqnarray*}
       \cS^\pm&:=&\{\si: \eta^\psi(x;\si)= \pm 1, \forall x \in \ell_+\mathbb{Z}\}.
\end{eqnarray*}
\end{defin}
Since we will describe the system in terms of the $\eta$ variables, it is convenient to take
$\Lambda$ as an $\ell_+$-measurable interval (i.e. $\Lambda=[h, k)\cap \mathbb{Z}$ for $h, k \in \ell_+\mathbb{Z}$, $h<k$).
For notational convenience we also take it centered at $0$ in the statement below.

\noindent Our main theorem is:

\begin{thm}
\label{thm:1}
{Let $\ell_-=\delta_-/\gamma,\ell_+=\delta_+/\gamma$ with $\delta_-,\delta_+$ chosen as in \eqref{eq:deltachoice}.
There exists a positive constant $\bar b$ such that if $\bar\beta(\lambda)$ is defined as the solution  of $\la\beta m_{\beta}^2=\bar b$, then
for any $\beta >\bar\beta(\la)$, there exists $\ga_0=\ga_0(\beta,\lambda)$ positive so that for all $\ga<\ga_0$ and all $\bar \si \in \cS^+$:}
\begin{eqnarray}
\label{eq:th1-a}
\mu_\La(\eta^{\psi}(0)\ne 1|\bar \si)&\le & e^{-\beta c'_\beta \psi^3 \ell_-}
\\
\label{eq:th-b}
\mu_\La(\eta^{\psi}(0)= -1|\bar \si)&\le & e^{-\beta\tilde J \ga^{-1}},
\end{eqnarray}
for some $c'_\beta>0$ and $\tilde J>0$ both depending $\beta$ ($\tilde J$ is the function given in (\ref{eq:JJJJ}). The same hold when $\bar\si=+\underline{1}$, provided $|\Lambda| \ge \exp{(2/\gamma)}$.
\end{thm}

The statements in Theorem \ref{thm:1} for $\bar\si\in \mathcal{S}^+$ are proven at the end of section \ref{sec:proofthm1} and the extension to
$\bar \si =+\und{1}$ is proven at the end of section \ref{sec:basic-estimates}.
\vskip 0.5cm

\noindent {\bf Remark.}  Theorem \ref{thm:main} follows at once from Theorem \ref{thm:1}.

\vskip 0.5cm \noindent
\section{Proof of Theorem \ref{thm:1}}
\label{sec:proofthm1}

\noindent The proof of Theorem \ref{thm:1} is obtained by a Peierls contour argument. We first consider the case of
$\sigma \in \mathcal{S}^+$, and at the end we discuss how to adapt the estimates to the case when
$\bar \si=+\underline{1}$.

\vskip 0.5cm
\subsection{Triangles and rectangles}
\label{sub:triangles}
\vskip .5cm \noindent
Given any external configuration $\bar \si\in \cS^+$, we associate to each coarse grained configuration $\eta_\La$ in the volume $\La$
a configuration of ``triangles" and ``rectangles", whose definition is a natural extension of the one given in \cite{CFMP}. The triangles
arise from a geometric procedure (in a plane containing our one-dimensional system) to determine the connection between two interface
points, marking a plus or a minus region. (This is also analogous to what happens for usual contours in dimension $d\ge2$, where the ``natural"
definition of connection is also appropriate to describe energy fluctuations.)

\noindent We start by setting variables  $\Theta(h)$ which work as ``phase indicators" on the coarse grained lattice. For
$h \in \ell_+\mathbb{Z}\setminus\Lambda$, set $\eta(h)=\eta(h,\bar\sigma)=+1$ and then define for $h \in\ell_+\mathbb{Z} \cap \Lambda$
\begin{eqnarray*}
\Theta(h)= \left\{
                \begin{array}{ll}
                  -1, & \hbox{if } \eta(h-\ell_+)=\eta(h)=\eta(h+\ell_+)=-1 \\
                  +1, & \hbox{if } \eta(h-\ell_+)=\eta(h)=\eta(h+\ell_+)=+1 \\
                   0, & \hbox{otherwise.}
                \end{array}
              \right.
\end{eqnarray*}
An $\ell_+$-measurable interval $[h,k)$  is called {\em ``almost positive" }
if
\begin{equation*}
\Theta(h)=\Theta(k-\ell_+)=+1 \;\;\; \text{and } \Theta(i)\ne -1 \forall i \in \ell_+ \mathbb{Z}, h<i<k-\ell_+.
\end{equation*}
Notice that $\Theta(\cdot)=0$ is allowed inside an almost positive interval; in particular the magnetization over such an interval might
be negative. {\em Almost negative } intervals are defined analogously.

\begin{defin}[rectangles]
The rectangles, denoted by the letter $Q$, are defined as the $\ell_+$-measurable intervals
that correspond to maximal (non-null) runs of $\ell_+$-blocks where $\Theta=0$.
\end{defin}

\noindent {\bf Remark.} A rectangle of size less than $3\ell_+$ can occur only as a
set of two consecutive $\ell_+$-blocks with $\eta=-1$ in one block, and $\eta=+1$ in the other one. An isolated
$h$ for which $\Theta(h)=0$ is not possible.

\vskip .5cm \noindent
Therefore, a configuration $\Theta$ can be regarded as a sequence of maximal ``almost positive"
and ``almost negative" intervals  separated by some  special rectangles
that mark the transition from an interval with a given sign to the next one of opposite
sign; such rectangles are then called {\em interface intervals}.
Each $\Theta$-configuration will be represented in terms of ``triangles" and ``rectangles",
the correspondence being bijective once the boundary conditions are fixed. As in \cite{FS} and \cite{CFMP}, our construction
is based on suitably coupling together pairs of interface points. To this end we will use the criterion
of minimal distance, which is made geometrically intuitive through a graphical representation where each
spin configuration is mapped into a set of triangles and rectangles. The endpoints of the triangles will
be pairs of suitable coupled interface points.

The precise location of an interface point is immaterial; for convenience we choose, for each
$i \in \mathbb{Z}$, a point $r(i)$ in each interval $[i, i+1/100]$ with the property that for any four distinct integers $i_j$, $j=1,..,4$, $|r(i_1)-r(i_2)|\ne |r(i_3)-r(i_4)|$.  This choice is done once for all. If $[h, k)$ is an interface interval,
the points $r_h$ and $r_k$ are defined as the corresponding {\em interface points}, and considered to be paired in the
the following construction.

The construction of the triangles is slightly more complicated.
We start by attributing colors (blue and red) to each pair of interface points. For an interface interval $[h,k)$:
the point $r_h$ is red (blue)
if 
{$\Theta(h-\ell_+)=+1$}(-1, respectively), in which case $r_k$ will be blue (red, respectively)
corresponding to the fact that 
{$\Theta(k)=-1$}(+1, respectively).

We  let each interface point evolve into a trajectory of the same color, represented in the $(r,t)$ plane by the
line $r\pm t$ or  $t\ge 0$. The choice between the two  directions  of the trajectory is made
in such a way that each red line (blue line) projects its shadow on the (contiguous) almost positive (almost negative, resp.) interval.
We have thus, unless $\Lambda$ is almost positive (i.e. $\Theta(h)\neq -1$ for all $h$), a bunch of growing connected-lines
each one emanating from an interface point. Red lines ignore  blue lines, and viceversa. When two lines of the same color
collide they stop growing and the line which corresponds to the paired interface point (of opposite color) is also canceled.
In the meantime, all the other lines keep growing. Our choice of the location of the interface points ensures that collisions
occur one at a time so that the above definition is unambiguous.

The process described above will stop in a finite time $t<|\La|$, giving rise to triangles.
In fact  the collision of two points is represented graphically in the $(r,t)$ plane by a
triangle whose basis is the line joining the two interface points and whose sides are
      the two lines which meet at the time of collision.
      Triangles will be usually denoted by $T$.
%
%

\begin{defin}[Triangle]
If  $[r_i, r_j]$ is the basis of a triangle in the above
construction, the $\ell_+$-measurable interval $T=[i,j) \cap \mathbb{Z}$ is called {\em triangle}.
\end{defin}

\noindent {\bf Remark.}
The neighboring external $\ell_+$--blocks to the left and to the right of a triangle have equal $\eta$-value, $+1$ or $-1$. This common sign
is set as the {\em sign} of the triangle. Notice that to the left and to the right of a triangle $T$ there are always at least two contiguous $\ell_+$-blocks where $\Theta=0$.

 Given any $\bar \si\in \cS^+$ and $\eta_\Lambda\neq +\und{1}$, we have defined the $\Theta$ configuration and
represented it as a collection  $(\und T,\und Q)=(T_1,..,T_n,Q_1\dots,Q_m)$ of triangles and rectangles. When
$\eta_\Lambda=+\und{1}$, we have an empty configuration, hereby denoted by $\emptyset$.

\vskip .5cm
Recalling the definition of $m^{\ell_-}(\si)$ in \eqref{def:m-si} we may set:

\begin{defin}[$\und S(m)$]
For any given boundary condition $\bar \si$, and for any $m=m^{\ell_-}(\si)$,
let $\und S(m;\bar \si)=(\und T(m), \und Q(m))$ denote
the configuration of triangles and rectangles that correspond to $\eta_\Lambda(m)$.
\end{defin}

\begin{defin}[distance]
\label{defin:distance}
For $A$ and $B$ non-empty subsets of $\mathbb{R}$,  $d(A,B)$ denotes the usual
distance between the two sets. If $Q_1$ and $Q_2$ are rectangles, we set $D(Q_1,Q_2)=d(Q_1,Q_2)$. On the other hand, if
$T_1$ and $T_2$ are triangles, let $D(T_1,T_2)= d(\e(T_1), \e(T_2))$, where $e(T)$ denotes the set of extremal points of the
interval (in $\mathbb{Z}$) which gives $T$. Finally, when $T$ is triangle and $Q$ a rectangle we set $D(T,Q)=d(\e(T), Q)$.
\end{defin}
\noindent Notice that by construction:
\begin{eqnarray}
\label{compTT}
&&D(T_1,T_2)\ge \min\{|T_1|, |T_2|\},
 \\
\label{compQQ} &&D(Q_1, Q_2)\ge \ell_+,
\end{eqnarray}
while between a triangle $T$ and a rectangle $Q$ of the same configuration
one of the two following relation holds:
\begin{eqnarray}
&&T\cap Q=\emptyset \;\;\text{and} \;\;  D(T,Q)\ge 1
\nn \\
&&\mbox{or}
\label{compTQ}
\\
&&Q\subset T \;\;\text{and} \;\;D(T,Q)\ge \ell_+.
\nn
\end{eqnarray}
In words, a rectangle can be attached to a triangle only externally.
\vskip .5cm \noindent

\begin{defin}[compatibility]
A given configuration of triangles and rectangles $(\und T,\und Q)$
(always related to $\ell_+$--measurable blocks) is called {\em compatible}
if for any couple of elements \footnote{We abuse language here.} of  $(\und T,\und Q)$  \eqref{compTT} \eqref{compQQ}, \eqref{compTQ} hold,
and moreover any rectangle $Q$ in this configuration satisfies $|Q|/\ell_+\ge 2$; moreover, when $|Q|=2\ell_+$ then $Q=[h,h+2\ell_+)$
for some $h \in \ell_+\mathbb{Z}$ with $\eta(h)\eta(h+\ell_+)=-1$.

In the sequel we denote by $\und S=(\und T,\und Q)$ a set of
compatible elements, and by the letter $S$ a generic element
of such a configuration $\und S$. We say that two sets of compatible
elements $\und S_1, \und S_2$ are compatible (and we use the notation $ \und S_1\sim \und S_2$)
if $\und S_1\cup \und S_2$ is a compatible configuration.
\end{defin}

\noindent {\bf Remark} \eqref{compTT} allows the
possibility that (for  compatible)  $Q_1,T_1,T_2$ :
$T_1\subset T_2$ or $Q_1\subset T_2$ (bur not viceversa).

\begin{defin}
\label{defin:triangles1}
For the external condition $\bar\si =+\und{1}$ we still define the triangles and rectangles in $\La$ as if
the boundary condition were in $\cS_\La^+$, namely ignoring the interface points on the boundary.
\end{defin}

\vskip 0.5cm \noindent
\subsection{Contours}
\label{sub:contours}

As before we first consider the case $\bar \si\in \cS_\La^{+}$.

Our aim is to apply the Peierls argument to our system. Following \cite{CFMP}, for any given configuration
$\und S\equiv (\und T, \und Q)$ on $\La$, we define a partition of $\und S$ by suitably grouping its ``elements" (or constituents) $S$
\footnote{When not needed to distinguish between triangles and rectangles, we use the notation $S$ to denote a generic
constituent of the configuration $\und S$. Also we slightly abuse language here.} in {\em contours} $\Ga^{(i)}$, which
should be sufficiently separated from each other, so as to exhibit a weak dependence. This grouping procedure
is obtained by an algorithm  that creates an hierarchical network of connections, that at the top level identifies
groups of connected constituents, namely the contours.

The algorithm $ \cR (\und S)$ on $\{\und S\}$ which
associates uniquely to any (compatible) configuration $ \und S $  a configuration of contours $\und \Gamma= \{ \Ga_j\}$
is the  same defined in  \cite {CFMP}. Here we quote only its properties, before which we need
the following notation:

We denote by $T(\Ga)$ the smallest interval which contains all the
elements of the contour, its right and left endpoints are denoted by $x_{\pm} (\Ga)$. We set:
\begin{eqnarray*}
  |\Ga|:= \sum_{S\in \Ga} \frac{|S|}{\ell_+}\,.
\end{eqnarray*}

\noindent {\em Properties of $\cR(\und S)$}

\noindent {\bf P.0}  {\it   Let $ \cR (\und T, \und Q) = ( \Ga_1, \dots, \Ga_n)$, $ \Ga_i= \{ T_{q,i} Q_{p,i}, 1 \le q\le k_i, 1 \le p \le m_i\}$,
then $\und T=  \{ T_{q,i}, 1\le i \le n, 1 \le q\le k_i\}$, $\und Q=  \{ Q_{p,i}, 1\le i \le n, 1 \le p\le m_i\}$.}

\noindent {\bf P.1}  {\it Contours are well separated from each other.}  All pairs $ \Ga \neq \Ga'$
 verify, for a suitable constant $\C$:
\begin{eqnarray}
\label{SS1}
D (\Ga, \Ga'):= \min_{S \in \Ga, S' \in \Ga'} D ( S,S') > \C \ell_+\min \left\{ |\Ga|^3, |\Ga'|^3\right \}.
\end{eqnarray}
Condition (\ref{SS1}) allows $T (\Ga) \cap T (\Ga') \ne \emptyset $ in which case either
$T(\Ga)\subset T(\Ga')$ or $T(\Ga')\subset T(\Ga)$; moreover, supposing for instance
that the former case is verified, (in which case we call $\Ga$ an inner contour), then for any
element $ S'_i \in \Ga'$, either $T(\Ga)\subset S'_i$ or  $T(\Ga)\cap  S'_i= \emptyset $  and $|\Ga|^3<|\Ga'|^3$.

\noindent{\bf Remark.} The constant $\C$ has to be taken large enough; in particular we need  $\C m^2_\beta>10$ in the statement of Proposition \ref{prop:m3.1}.

\noindent {\bf P.2}  {\it Independence.}  Let $\{ \und S^{(1)}, \dots, \und S^{(k)}\}$
be $k$ compatible configurations (where $k>1$ ) of rectangles and triangles;    $\cR ( \underline S^{(i)}) = \{ \Ga_j^{(i)}, j=1,\dots, n_i\}$
 the contours of the configurations $\underline S^{(i)}$.
If any two distinct $ \Ga_j^{(i)}$ and
 $\Ga_{j'}^{(i')}$ satisfy {\bf P.1},
$$\cR ( \underline S^{(1)}, \dots, \underline S^{(k)} )= \{  \Ga_j^{(i)}, j=1,\dots, n_i; i=1, \dots, k \}. $$
  It was proved in \cite {CFMP} that not  only  {\bf P.0},  {\bf P.1} and   {\bf P.2} can be actually
 implemented by some algorithm $\cR$, but that such algorithm is unique and therefore there is a
bijection between (compatible) configurations of triangles and rectangles and the contours. It is easy  to convince  oneself that the above mentioned proof
holds also for our sets $\{(\und T,\und Q)\}$ when rectangles are also present, cf. section $3.1$ of \cite{CFMP}.

\medskip
\noindent

The following two propositions summarize the relevant properties of the contours:

\begin{prop}
\label{prop:m3.1}
Let $X_\Ga$ denote the event that $\Ga$ is a contour in the configuration $\und{\Ga}(\sigma)$ on $\La$.
Then, uniformly in $\La, \bar\si\in \cS^+_\La$, and for $\gamma$ small enough:
\begin{equation}
\label{eq:prop.m3.1}
\mu_\La(X_\Ga|\bar \si)\le \prod_{T \in \Ga}(e^{-2\la\beta m_\beta^2a_\beta\ln(|T|\gamma)}
e^{-\frac{\beta}{\gamma} (\tilde J-5\tilde\lambda\ln 5)})\prod_{Q \in \Ga} (e^{-{\beta}\epsilon\ell_-\frac{|Q|}{7\ell_+}}3^{\frac{|Q|}{\ell_+}}),
\end{equation}
where $a_\beta=(1-10/(\C m_\beta^2))(1-\psi/m_\beta)^2$, with $\C$ as in \eqref{SS1}, $\tilde J$  and $\epsilon$ are given by (\ref{eq:JJJJ})  and (\ref{eq-prop:mu-rectangles}) respectively.
In particular, if the scales are chosen as in (\ref{eq:deltachoice}), then
\begin{equation}
\label{eq:prop.m3.1b}
\mu_\La(X_\Ga|\bar \si)\le \prod_{S \in \Ga}e^{-[b_\beta\ln(|S|\gamma)+c(\gamma)]},
\end{equation}
where $b_\beta=2\la\beta m^2_\beta a_\beta$
and $c(\gamma)=c(\gamma,\beta) >0$ tends to $\infty$ as $\ga \to 0$.
\end{prop}

\noindent The proof is given in the next section.
\vskip 1cm
\begin{prop}
\label{prop:m3.2}
For all $b$ sufficiently large, $c>\ln 2$,  $m\ge 3$:
\begin{equation}
\label{eq:prop.m3.2}
\sumtwo{\Ga\ni 0}{|\Ga|=m}\prod_{S\in \Ga} e^{-b[\ln(|S|\gamma)]-c}\le 2 m e^{-b\ln m-(c-\ln 2)}
\end{equation}
\end{prop}

\begin{proof}
If we do not distinguish between $Q$ and $T$ while counting the contours, the current entropy estimate boils down to
Theorem 4.1 in \cite{CFMP},  which relies only on the properties {\bf P.0}-{\bf P.2} of the contours.
The only difference is an extra combinatorial factor,
since each $S \in \Gamma$ can be a $Q$ or a $T$. This amounts
to have an extra factor $2$ inside the product over $S$,
which can be easily controlled if $c > ln 2$
\end{proof}
\medskip

\noindent {\bf Remark.} A careful examination of the proof of Theorem 4.1 in \cite{CFMP}
(see also the appendices E and F there)
shows that the statement of previous proposition holds for any $b$ such that
\begin{equation}
\label{eq:sarava}
\sum_m m^6 {e^{-\frac{b(1-\rho)}{2} \ln m}} <  \frac{e^{c}}{2\varpi},
\end{equation}
where $\varpi$ is the parameter introduced in \eqref{SS1} and $\rho$ is any number arbitrary small.  This allows to prove our main theorem for
$\beta$ such that $\lambda \beta m_\beta^2 > 7$.

\medskip
\noindent {\em Proof of Theorem \ref {thm:1} for $\bar \sigma \in \mathcal{S}^+$.}
The definition of contours implies that
\begin{equation}
\label{eq:m4}
\mu_\La(\Ii_{\eta(0)\ne 1}|\bar \si)\le \sum_{\Ga\ni 0} \mu_\La(X_\Ga|\bar \si).
\end{equation}
From this we see that the proof of (\ref{eq:th1-a}) in Theorem \ref{thm:1} follows at once from propositions \ref{prop:m3.1}, \ref{prop:m3.2}.
The second inequality follows similar lines.

\vskip 0.5cm
\section{Proof of \eqref{eq:prop.m3.1} of Proposition \ref{prop:m3.1}}
\label{sec:proofof m31}

\noindent To exploit the mean field limit, we need to express the Gibbs measure in terms of the
coarse grained variables $m(\cdot)=m^{\ell_0}(\cdot,\sigma)\in\mathcal{M}_{0,\Lambda}$
introduced in section \ref{sec:coarsegraining}, plus an error term to be controlled
if $\gamma$ is small enough.  We write
\begin{eqnarray*}
Z_\La(m|\bar\si):=\sum_{\si_\La\colon m^{\ell_0}(\cdot,\sigma)=m(\cdot)}e^{-\beta H(\si|\bar\si)}=
e^{-\beta\ga^{-1}[F_\La(m|\bar m)+\cG_\La(m|\bar \si)]},
\end{eqnarray*}
where $F_\La(m|\bar m)$ and $\cG_\La(m|\bar \si)$ are defined as follows:
\begin{eqnarray}
\label{def:Fm}
F_\La(m|\bar m)&:=&\delta_0\sum_{x\in \La\cap \ell_0\mathbb{Z}}f_\beta(m(x))
\\&& +\frac{(\delta_0)^2}{4\gamma}\sum_{x,y\in \ell_0\mathbb{Z}\cap \La}J_\gamma(|x-y|)(m(x)-m(y))^2\nn
\\&&+\frac{(\delta_0)^2}{2\gamma}\sum_{x \in\ell_0\mathbb{Z}\cap\La,y \in \ell_0\mathbb{Z}\setminus \La}J_\gamma(|x-y|)(m(x)-\bar m(y))^2\nn
\\&&+\frac{(\delta_0)^2}{2\gamma}\sum_{x \in\ell_0\mathbb{Z}\cap\La,y \in \ell_0\mathbb{Z}\setminus \La}J_\gamma(|x-y|)(\bar m(y))^2\nn
\end{eqnarray}
with $\bar m(x):=m^{\ell_0}(x;\bar\si)$, and $f_\beta$ the (limiting) mean field free energy
\begin{eqnarray*}
&&f_\beta(m)= -\frac12 m^2 -\frac{1}{\beta}S(m) \,\text{ with }\\\nn
&&S(m)=-\frac{1+m}{2}\ln\frac{1+m}{2}-\frac{1-m}{2}\ln\frac{1-m}{2}, \;\; m\in (-1,1),
\end{eqnarray*}

\begin{eqnarray}
\label{def:cG}
&&\cG_\La(m\mid\bar \si):=\frac{\delta_0}{\beta}\sum_{x \in \Lambda \cap \ell_0\mathbb{Z}} S(m(x)) +\frac{1-\|J_{\ga}\|_0}{2}\delta_0\sum_{x\in \La\cap \ell_0\mathbb{Z}}(m(x))^2 \\
&& -\frac{\ga}{\beta}\ln\bigg(\sum_{\si: m^{\ell_0}(\si)=m}
\exp\bigg\{-\frac{\beta}{2}\sumtwo{x, y\in\ell_0\Zz\cap \La}{i\in C_x^0,j\in C_y^0}
(J_\ga(|i-j|)-J_\ga(|x-y|))\si_i\si_j \nn \\
&& -\beta\sumtwo{x \in \ell_0 \mathbb{Z}\cap \La, y \in\ell_0\mathbb{Z}\setminus \La}{i\in C_x^0,j\in C_y^0}(J_\ga(|i-j|)-J_\ga(|x-y|))\si_i\bar\si_j \bigg\} \bigg), \nn
\end{eqnarray}
where
\begin{equation*}
\|J_{\ga}\|_0 :=\ell_0\sum_{i \in \mathbb{Z}}J_\gamma(i\ell_0).
\end{equation*}
Notice that $\|J_{\ga}\|_0 = 1 + O(\delta_0)$ and that $J^{(1)}$ does not contribute to the mean field limit, as
seen at once from (\ref{def:J}).
For any $\bar \si\in  \cS^+$  and $\und \Ga$ a compatible configuration of contours
in $\La$, we write:

\begin{eqnarray}
\label{def:hatH}
\hat H_{\bar\sigma}(\und\Ga)= -\frac{\ga}{\beta}\ln \sum_{m\in \cE_\La(\und \Ga)}
 e^{-\beta\ga^{-1} [F_\La(m|\bar m)+\cG_\La(m|\bar \si)]},
\end{eqnarray}
where $\cE_\La(\und \Ga)$ stands for set of possible profiles $m^{\ell_0}(\cdot)$ which give rise
to such configuration of contours. Hence,
\begin{eqnarray}
\label{DEF:H} Z_\La(\bar\si)=\sum_{\und \Ga}
e^{-\beta\ga^{-1}\hat H_{\bar\sigma}(\und\Ga)}
\end{eqnarray}
where the sum is over all the possible sets of compatible configurations
of contours in $\La$.

\noindent {\bf Notation.} Since $\bar\sigma \in \mathcal {S}^+$ is fixed in the derivation below and the estimates do not depend
on its value, we omit the subscript in $\hat H_{\bar\sigma}(\und\Ga)$ in the sequel.

The next lemmas summarize basic properties of $\hat H(\und\Ga)$. The first will be
proven in the next section, and the second follows easily from the same arguments
as in \cite{CFMP}.

\begin{lemma}
   \label{lem:m4.1}
Let $(\und T,\und Q)$ be a configuration in $\La$ with a unique contour $\Ga_0$. Then, for $\gamma$ small enough:
    \begin{equation}
\label{eq:lem-m4.1}
\hat H(\Ga_0)- \hat H(\emptyset)\ge W(\Ga_0):=
\sum_{Q\in \Ga_0} \frac{\delta_-\epsilon}{7}\frac{|Q|}{\ell_+}
+ \sum_{T\in \Ga_0} \left({2\tilde\la} (m_\beta-\psi)^2\ln(|T|\gamma)+\tilde J-5\tilde\lambda\ln5\right),
\end{equation}
where $\tilde J$ and $\epsilon$ are the same as in Proposition \ref{prop:m3.1}.
\end{lemma}

\begin{lemma}
   \label{lem:m4.2}
Let $\und \Ga\cup \Ga_0$ be a compatible configuration of contours, then:
    \begin{equation}
\label{eq:lem-m4.2}
\hat H(\und\Ga \cup\Ga_0)- \hat H(\und \Ga)\ge W(\Ga_0)\left(1-\frac{10}{\C m_\beta^2}\right)
\end{equation}
$W(\Ga_0)$ the same as defined in Lemma \ref{lem:m4.1} and $\C$ is the same constant appearing  in \eqref{SS1}.
\end{lemma}
\begin{proof}
The proof is based on the property {\bf P.1} in the definition of contours, which allows to neglect the
interactions between contours, when $\C$ is chosen sufficiently large so that $\C m_\beta^2$ is larger than 10.
For the explicit estimates we refer to section $3.2$ in \cite{CFMP}.
\end{proof}

\vskip .5cm \noindent
\begin{proof}[Proof of proposition \ref{prop:m3.1}]
If we notice that uniformly in
$\La, \bar\si\in \cS^+_\La$, $\Ga_0$:
\begin{eqnarray*}
\mu_\La(X_{\Ga_0}|\bar \si)=\frac{\sum_{\und \Ga\sim \Ga_0}e^{-\beta \ga^{-1} \hat H(\und\Ga\cup \Ga_0)}}{Z_\Lambda(\bar\sigma)}\le
e^{-\beta \ga^{-1}\inf_{\und \Ga\sim \Ga_0}[ \hat H(\und\Ga\cup \Ga_0)-\hat H(\und \Ga)]},
\end{eqnarray*}
the proof of \eqref{eq:prop.m3.1}
follows from Lemma \ref{lem:m4.2} given that the parameter $\C$ has been chosen as above.
\end{proof}

\vskip 0.5cm \noindent

\noindent {\em Strategy of the proof of Lemma  \ref{lem:m4.1}}

We are assuming that there is a unique contour $\Ga$. It consists in a set of triangles and rectangles $\Ga=\{S_i\}_{i=1,\dots,n}$. We will proceed in the estimate
of Lemma \ref{lem:m4.1}
iteratively ``removing" elements one at a time. Namely re-writing the l.h.s of \eqref{eq:lem-m4.1} as:
\begin{eqnarray}
\label{eq:via}
[\hat H(\Ga)- \hat H(\emptyset)]&\ge& [\hat H(\Ga)- \hat H(\emptyset;\cA(\Ga)]
\\ \nn
&\equiv&
\sum_{i=1}^n [\hat H(\Ga \setminus \cup_{j<i}S_j; \cA(\cup_{j<i}S_j))-
\hat H(\Ga\setminus \cup_{j\le i}S_j;\cA(\cup_{j\le i}S_j))]
\\ \nn
&\equiv&
\sum_{i=1}^n W_\Ga(S_i)
\end{eqnarray}
where, for any subset $\cD\subset \cM_{0,\La}$ we define:

\begin{eqnarray}
\label{def:hatHxA}
\hat H(\und\Ga;\cD):=-\frac{\ga}{\beta}\ln \sum_{m\in \cE_\La(\und \Ga)\cap \cD}
 e^{-\beta\ga^{-1} [F_\La(m|\bar m)+\cG_\La(m|\bar \si)]},
\end{eqnarray}
and we will define in the sequel  suitable choices for the sets $\cA(\cup_{j<i}S_i)$ depending
only on the support of $\cup_{j<i}S_i$, and such that $\hat H(\emptyset)\le \hat H(\emptyset;\cA(\Ga))$.

\noindent\begin{defin} \label{def:isoltedQ}
A rectangle $Q$ in a configuration $(\und T, \und Q)$ is called {\em ``isolated rectangle"}
if $D(Q,T)\ge \ell_+$ for any triangle $T$ in $(\und T, \und Q)$.
Otherwise we say that it is an {\em ``attached rectangle"}.
\end{defin}


In the next section we will prove these basic estimates, namely the lower bound
for the two types of elements $S$, labeled as outlined above.

\vskip .5cm \noindent
\section{Basic estimates}
\label{sec:basic-estimates}

\noindent Let  $\Ga= \{S_i\}_{i=1,\dots n}$ be a contour, which we assume to be the only contour of the whole configuration in $\La$.

\noindent {\em Choice of the scales.} The choice of $\delta_0,\delta_-,\delta_+$ as functions of $\gamma$ is not  strict, but we fix
here some suitable values (in terms of $\ga$) as follows:
\begin{eqnarray}
\label{eq:deltachoice}
\delta_0=\ga^{1/2} \hskip1cm \delta_-=\frac{1}{\ln\ga^{-1}} \hskip1cm \delta_+=\ga^{-1/2}\frac{1}{(\ln\ga^{-1})^3}.
\end{eqnarray}

Given $\epsilon>0$, the following inequalities hold for $\gamma$ sufficiently
small:
\begin{eqnarray}
&& \delta_-\epsilon> \frac{\delta_+}{\ell_0}
\\
&& \delta_-\epsilon>\tilde \la \ln \delta_+
\\
&& \delta_-\epsilon>\frac{\delta_+}{\ell_0} \ln\ell_0
\\
&& \delta_-\epsilon>\delta_0 \delta_+.
\end{eqnarray}
In particular, for $\ga$ sufficiently small and for any $n$, $n\delta_-\epsilon>\tilde \la \ln (n\delta_+)$.
(We use $\epsilon$ as in (\ref{eq-prop:mu-rectangles}).)

\noindent {\bf Remark.} In the following derivations,  $c, c^\prime, ...$ indicate constants whose value is not truly important (even if depending sometimes on the parameter $\beta$) and may change from place to place.
\vskip 0.5 cm.
\centerline{\em The  first basic estimate}

Following the strategy outlined in  \eqref{eq:via} when $S_1$ is an isolated rectangle and
$\cA(S_1)=\cM_{0,\La}$, the whole set of possible profiles $m$. Hence
$\hat H(\Ga;\cM_{0,\La})\equiv \hat H(\Ga)$.

\begin{lemma}
   \label{lem:mQ}
There exist a constant $c'_\beta$ so that for any isolated rectangle $Q$ and any contour
$\Ga$ which contains $Q$:

    \begin{equation}
\label{eq:lem:mQ}
\hat H(\Ga)-\hat H(\Ga\setminus Q)\ge \frac{|Q|}{6\ell_+} c'_\beta \delta_-\psi^3.
\end{equation}\end{lemma}

{\bf Remark 1}\label{remark2}
Let $Q=[h, k)$ with $h, k \in \ell_+\mathbb{Z}\cap \La$ be an isolated rectangle.
By definition $\Theta(j)=0$ for any $j\in Q\cap\ell_+\mathbb{Z}$ and $\Theta(h-\ell_+)=\Theta(k)\neq 0$.
Let us assume, without loss of generality that this common value is $+$, the
opposite case being completely analogous. In this case the following
holds for the  $\eta$-variables:
\begin{equation*}
\eta(h-2\ell_+)=\eta(h-\ell_+)=\eta(h)= \eta(k-\ell_+)=\eta(k)=\eta(k+\ell_+)=+1
\end{equation*}
and  $|\eta(i-\ell_+)+\eta(i) +\eta(i+\ell_+)|\le 2$ for any $i\in [h,k-\ell_+) \cap \ell_+ \mathbb{Z}$,
i.e., that there are no three consecutive $C^+$ blocks where $\eta$ has the same sign.

{\bf Remark 2} \label{remark22}
In each $C^+_h$ block with $\eta(h)=0$, we partition the set of possible configurations $m=m^{\ell_0}(\cdot,\sigma)$ as follows:
\begin{itemize}
  \item[(a)]
$\cA_{h}:=\{m^{\ell_-}: \eta(h)=0, \, \sup_{x\in C^+_h\cap \ell_-\mathbb{Z}}\big(\big||m^{\ell_-}(x,\sigma)|-m_\beta\big|\big)>\psi\}$
  \item[(b)]
$\cB_{h}:=\{m^{\ell_-}: \eta(h)=0, \, \sup_{x\in C^+_h\cap \ell_-\mathbb{Z}}\big(\big||m^{\ell_-}(x,\sigma)|-m_\beta\big|\big)\le\psi\}$.
\end{itemize}

We then denote by
\begin{eqnarray}
\label{def:mua}
\epsilon_a &:=& \inf_{m: m^{\ell_{-}}(m)\in \cA_{h}} F_{C^+_h}(m|m) - F_{C^+}(m_\beta|m_\beta)
\\
\label{def:mub}
\epsilon_b &:=& \inf_{m: m^{\ell_{-}}(m)\in \cB_{h}} F_{C^+_h}(m|m) - F_{C^+}(m_\beta|m_\beta).
\end{eqnarray}

 \vskip .5cm \noindent

{\bf Remark 3}\label{remark3}
Notice that, since the interaction energy is positive
(so that enlarging the size of region does not lower the minimum),
the free energy of two contiguous cubes with opposite signs of the variable
 $\eta$ is bigger than or equal to $\epsilon_b$

\vskip .5cm \noindent

{\bf Remark 4}\label{remark4}
By the previous remarks $1,2,3$ we have
that the minimal free energy of a rectangle composed by $n$  $C^{+}$-blocks
is at least $\min\{\epsilon_a,\epsilon_b\} \frac{n}{3}$. In appendix \ref{prop:mu-rectangles} the following proposition
will be proved:

\vskip .5cm \noindent
\begin{prop}
\label{prop:mu-rectangles}
Let $\epsilon_a,\epsilon_b$  be defined as in the Remark 2 above. Then:

\begin{eqnarray}
\label{eq-prop:mu-rectangles}
\epsilon:= \frac{\min\{\epsilon_a, \epsilon_b \}}{\delta_-}> c'_\beta \psi^3
\end{eqnarray}
where $c'_\beta$ is a positive constant.
\end{prop}

\begin{proof}[Proof of Lemma \ref{lem:mQ}]
{\bf Notation.} For $m\in [-1,1]$, $(m)_\ga$ denotes the best approximation of $m$ in $\cM_{0}$.

In the previous notation, i.e. writing $Q=[h,k)$ for the isolated rectangle as before, we
set (for any configuration $m \in \mathcal {M}_{0,Q}$ )
\begin{itemize}
\item

\begin{eqnarray}
\label{def:m-errico}
\tilde m(x;m)=\left\{
                \begin{array}{ll}
                  m(x), & \hbox{if } \; x\notin \cB(Q)\\
                  (\phi_{\cB(Q)}(x;m_{[\cB(Q)]^c}))_\ga & \hbox{if } \; x\in \cB(Q)
                \end{array}
              \right.
\end{eqnarray}
with $\cB(Q):=\{x\in Q\colon d(x,Q^c)<\ell_+\}$
and $\phi_A(x;m_{A^c})$ the function defined in Lemma
\ref{thm:errico}, that has the following property stated in corollary \ref{corol:errico}:
\begin{eqnarray}
\label{eq:errico}
&&|\phi_\Delta(x)-m_\beta|<e^{-c\frac{1}{3}\delta_+} < 1/{\ell_0} \\ \nn
&&\hskip1cm \forall x\in \Delta:=\{x\in C^+_h:
d(x, [\cB(Q)]^c)>\ell_+/3\}
\end{eqnarray}
where the second inequality in \eqref{eq:errico} follows at once from the choice
of the scales (see \eqref{eq:deltachoice}) for $\gamma$ small.

\medskip

\item
$m^*\equiv m^*_{Q}(m)$, defined as follows:
\begin{eqnarray}
\label{def:mstarQ}
m^*_{Q}(x;m)\equiv m^*(x):=\left\{
          \begin{array}{ll}
            \tilde m(x;m), & \hbox{if }  x\notin \hat Q\\
            s (m_\beta)_\ga, & \hbox{if }  x\in \hat Q
          \end{array}
        \right.
\end{eqnarray}
where
\[
\hat Q:=\{x\in Q\colon d(x,[Q]^c)>\ell_+/3  \}
\]
and $s\in\{- 1,+1\}$ is chosen equal to  $\pm 1$  depending on the sign of the neighboring blocks.

\end{itemize}

\begin{eqnarray}
\label{eq:infQ}
\hat H(\Ga)-\hat H(\Ga\setminus Q)&=&
-\frac{\ga}{\beta}\ln \frac{\sum_{m\in \cE(\Ga)}e^{-\beta\ga^{-1} [F(m|\bar m)+ G(m|\bar \si)]}}
{\sum_{m\in \cE(\Ga\setminus Q^o)}e^{-\beta\ga^{-1}[ F(m|\bar m)+ G(m|\bar \si)]}}
\\ \nn
&\ge&
\inf_{m\in\cE(\Ga) }
\{[F_\La(m|\bar m)+G_\La(m|\bar \si)]-[F_\La(m^*|\bar m)+G_\La(m^*|\bar \si)]\\
&&\hskip3cm 
-  \frac{\ga}{\beta}\frac{|Q|}{\ell_0}\ln \ell_0.
\end{eqnarray}

By lemma \ref{thm:errico} and corollary \ref{corol:errico} we have that:
\begin{eqnarray*}
F_\La(m|\bar m)-F_\La(\tilde m(m)|\bar m)\ge -2c\tilde \la\ln(\delta_+)-\frac{c\delta_+}{\ell_0}
\end{eqnarray*}
so that:

\begin{eqnarray*}
&&\hat H(\Ga)- \hat H(\Ga\setminus Q )\ge
\\
&& \ge
\inf_{m\in\cE(\Ga) }
\{[F_\La(\tilde m(m)|\bar m)+\cG_\La(m|\bar \si)]-[F_\La(m^*|\bar m)+\cG_\La(m^*|\bar \si)]\\
&&\hskip3cm 
-2c\tilde \la\ln(\delta_+)
-  \frac{\ga}{\beta}\frac{|Q|}{\ell_0}\ln \ell_0-c\frac{\gamma|Q|}{\ell_0}.
\end{eqnarray*}
\vskip 1cm \noindent

By remarks $4$, Proposition \ref{prop:mu-rectangles}, equation \eqref{eq:errico}
 and the estimate on $\cG_\La(m|\bar \si)-\cG_\La(m^*|\bar \si)$
proved in appendix \ref{app:appA0} 
we get:
\begin{eqnarray}
\label{eq: Q01}
&&\inf_{m\in\cE(\Ga)}
\{[F_\La(\tilde m(m)|\bar m)+\cG_\La(m|\bar \si)]-[F_\La(m^*|\bar m)+\cG_\La(m^*|\bar \si)]>
\\&&\hskip2cm >\frac{|Q|}{3\ell_+} c'_\beta\psi^3 \delta_- - c\frac{\gamma|Q|}{\ell_0}-\frac{\ga}{\beta}\frac{|Q|}{\ell_0}\ln \ell_0-{c'}\tilde \la\ln(|Q|\gamma)
-c\delta_0\ga|Q|\nn.
\end{eqnarray}
Recalling the choice of $\delta_0,\delta_-,\delta_+$ \eqref{eq:deltachoice}, for $\ga$
sufficiently small we get the result \eqref{eq:lem:mQ}.

\noindent {\bf Remark} This proof does not exploit the hypothesis that the whole configuration consists
of a single contour. It is actually true uniformly over all possible configurations compatible with $Q$.
\end{proof}

\vskip 0.5cm \noindent

\centerline{\em The second basic estimate}

\begin{lemma}
\label{thm:k}
Let us assume that $\Gamma$ is a configuration without isolated rectangles, with a smallest triangle $T$ and two attached rectangles,
 $Q_l, Q_r$ to the left and the right side of $T$. For all such configurations the following estimate holds true:
\begin{eqnarray}
\label{eq:Tbound}
\hat H(\Ga) &-&\hat H(\Ga\setminus [T\cup Q_l\cup Q_r]) \ge\bigg[\tilde J+ \frac{|Q_l|}{6\ell_+}\delta_-\epsilon + \frac{|Q_r|}{6\ell_+}\delta_-\epsilon\bigg]\nn\\&+&2\tilde \la (m_\beta-\psi)^2\ln(|T|\gamma)-c'\tilde \la\ln(|Q_r||Q_l|\ga^2)\\
&-&\ga\frac{|Q_l|+|Q_r|}{\ell_0}(\frac{1}{\beta}\ln\ell_0 +c)
-4c\tilde\la\ln (\delta_+)
-5\tilde\la \ln4-8\tilde\la K\delta_-\nn
\end{eqnarray}
where $\epsilon$ is given by (\ref{eq-prop:mu-rectangles}),
\begin{eqnarray}
\label{eq:JJJJ}
\tilde J:=\lim_{\La\to \Rr}\inf_{m_\La} F^0_\La(m_\La|-[m_\beta]_{\La_-^c},+[m_\beta]_{\La_+^c}) - F^0_\Lambda(m_\beta|m_\beta)
\end{eqnarray}
with the upper index in $F^0_\Lambda$ indicating the functional calculated for $\lambda=0$, $\La_-^c$ and $\La_+^c$ denote respectively the left (right) intervals of $\Lambda^c$, and $c, c', K$ are constants.
\end{lemma}

\begin{proof}
As in the previous proof, letting $\Delta=T\cup Q_l\cup Q_r$ we have
\begin{eqnarray*}
&&\hat H(\Ga) - \hat H(\Ga\setminus [T\cup Q_l\cup Q_r])
\\
&&\hskip 0.5cm \ge\inf_{m\in \cE (\Ga)}
\{[F_{\Delta}(\tilde m(m)|m \circ \bar m)+\mathcal{G}_{\Lambda}(m|\bar\sigma)]-
[F_{\Delta}(m^*(m)|m\circ\bar m)+\mathcal{G}_{\Lambda}(m^*(m)|\bar \sigma)]\} \\
&&\hskip 0.5cm - \frac{\ga}{\beta}\frac{|\Delta\setminus T|}{\ell_0}\ln \ell_0 -c\frac{|\Delta\setminus T|\gamma}{\ell_0}-4c\tilde \la \ln \delta_+,
\end{eqnarray*}
where $m\circ \bar m$ is the configuration which agrees with $m$ in $\La$ and with $\bar m=m^{\ell_0}(\cdot,\bar\sigma)$ outside
$\Lambda$, $\tilde m(m)$ is as in \eqref{def:m-errico}  with $\cB(Q)$ replaced by
$\cB(Q_l)\cup \cB(Q_r)$ and
\begin{eqnarray*}
 m^*_{T_0}(x;m):=\left\{
          \begin{array}{ll}
            m(x), & \hbox{if }  x\notin \Delta\\
            -m(x), & \hbox{if }  x\in  T\\
            (\phi_{B(Q_u)}(x))_\gamma, & \hbox{if }  x\in B(Q_u), \; u=l,r\\
 (sm_\beta)_\gamma, & \hbox{if }  x\in Q_u\setminus B(Q_u), \; u=l,r
          \end{array}
        \right.
\end{eqnarray*}
with $sm_\beta=\pm m_\beta$ according to the sign of the triangle.

The estimate for the contribution of $\cG$ is left to the appendix \ref{app:appA0}. We consider only
\begin{eqnarray*}
\inf_{m\in \cE (\Ga)}
\{F_{\Delta}(\tilde m(m)|\bar m)-
F_{\Delta}(m^*(m)|\bar m)\}&\ge&\inf_{m\in \cE (\Ga)}\{F_{\Delta\setminus T}(\tilde m(m)|\bar m)-
F_{\Delta\setminus T}(m^*(m)|\bar m)\}\\ &+&\inf_{m\in \cE (\Ga)}\{F_{T, \Delta^c}(\tilde m(m))-
F_{ T, \Delta^c}(m^*(m))\},
\end{eqnarray*}
where, for $A, B$ disjoint $\ell_0$-measurable intervals and $m \in \mathcal{M}_{0,\Zz}$:
\begin{eqnarray}
\label{def:Fauxiliar}
F_{A,B}(m)=\delta_0\sum_{x \in A \cap \ell_0\mathbb{Z}}f_\beta(m(x))+ \frac{(\delta_0)^2}{2\gamma}\sumtwo{x \in A \cap \ell_0\mathbb{Z}}{y \in B \cap \ell_0\mathbb{Z}}J_\gamma(|x-y|)(m(x)-m(y))^2.
\end{eqnarray}

Simple computations give the following estimates:
\begin{eqnarray*}
&&\inf_{m\in \cE (\Ga)}\{F_{\Delta\setminus T}(\tilde m(m)|\bar m)-
F_{\Delta\setminus T}(m^*(m)|\bar m)\}\ge \\
&&\hskip 3cm \max\left\{\tilde J, \frac{|Q_r|}{3\ell_+}\delta_-\epsilon\right\}+\max\left\{\tilde J, \frac{|Q_l|}{3\ell_+}\delta_-\epsilon\right\}
- c\tilde\la\ln(|Q_r||Q_l|\ga^2)\\
&&\hskip 3cm - \frac{(|Q_l|+|Q_r|)\gamma c}{\ell_0}
\end{eqnarray*}
\begin{eqnarray}
\label{eq:triangulo}
&&\inf_{m\in \cE (\Ga)}\{F_{ T,\Delta^c}(\tilde m(m)\circ \bar m)-
F_{ T,\Delta^c}(m^*(m)\circ\bar m)\}
\\\nn
&&\hskip 3cm\ge \tilde \la(m_\beta-\psi)^2\ln \frac{|T|^2}{|Q_r||Q_l|}-5\tilde \la\ln4-8\tilde\la K\delta_-.
\end{eqnarray}
The computations are carried out in  appendix \ref{app:triangulo}.
\end{proof}

\vskip .5cm \noindent

\noindent {\em Proof of the statements in Theorem \ref{thm:1} when $\bar\sigma=+\und{1}$}

Having defined the contours as if $\bar \eta \equiv\eta^{\psi}(\bar \si)=+\und{1}$, the difference in the
basic estimates occurs when the contour reaches the boundary of $\Lambda$. Since the external $\bar \si=+\und{1}$
favors the appearance of $\eta(h)=0$ close to the boundary, our previous estimates must be modified for contours
$\Ga$ such that $d(\Ga, \La^c)=1$.  In this case (and now we make explicit the boundary
conditions as subindex of $\hat H$), the estimate in \eqref{eq:lem-m4.1} has to be modified as follows:

\begin{equation}
\label{eq:piu}
\hat H_{+\und{1}}(\Gamma_0) - \hat H_{+\und{1}}(\emptyset) \ge W_{+\und{1}}(\Gamma_0) \ge W(\Gamma_0) - \frac{a(\beta,\lambda)}{\gamma}.
\end{equation}
for a suitable $a(\beta,\lambda)$ which can be taken less than 2. Since Proposition \ref{prop:m3.2}
still holds, the contribution of the contours that contain the origin and touch the boundary is easily controlled, when
$|\Lambda| \ge \exp{\left(\frac{2}{\gamma}\right)}$. This proves the statement.
\vskip 0.5 cm

\section{Final comments.}

\noindent Of course, by the spin flip symmetry, an analogous statement to Theorem \ref{thm:1} holds if $\bar \si \in \mathcal S^-$ (and $\bar\si=-\und{1}$ respectively). As a consequence,
and taking sequential limits $\mu_{\Lambda_n}(\cdot| \bar\si)$ with $\bar\si\in \mathcal{S}^\pm$, we obtain, for $\beta > \bar \beta(\lambda)$, at least two distinct Gibbs measures. It is also known (see \cite{FVV}) that in the present context all Gibbs measures are translationally invariant. One wonders if these limits do not depend on the specific choice of $\bar \si$ and coincide respectively with $\mu^{\pm}_{\beta,\gamma}$ of Theorem \ref{thm:main}. Using e.g. the
relativized Dobrushin criteria (see \cite{e-book}, \cite{BKMP1}) and cluster expansion techniques one should be able to prove this for $\beta$  large as in this paper.
\medskip

\noindent The techniques used in this paper can also be applied to slower decaying interactions (e.g. $\frac{ \la}{r^{2-\al}}$ for $\al<(\frac{\ln 3}{\ln 2}-1)$, see \cite{CFMP}) and  boundary conditions $\bar \sigma\in \cS^\pm$.
In this case (cf. \eqref{def:J}, \eqref{def:JJ}) $\tilde \lambda =\lambda \gamma^{1-\al}$ and $b_\beta$ in \eqref{eq:prop.m3.1b} tends to infinity as $\gamma \to 0$ for any $\beta>1$. Therefore the Peierls bound holds for any $\beta>1$ and $\gamma$ sufficiently small,
 $\beta_c(\gamma)\to 1$ as $\ga \to 0$ and the analogue of \eqref{eq:m3} is valid for any $\beta>1$. The case with boundary conditions $\bar \si\equiv +1$ is not contained because inequalities \eqref {eq:piu} are not valid.

\medskip

\noindent {\bf Acknowledgements:} We thank E. Presutti for suggestions and many discussions on Kac models. The authors thank the following institutions for their warm hospitality which made possible for us to carry out this work: {\em CBPF} and {\em IMPA}, in Rio de Janeiro, the Universities {\em La Sapienza} and {\em Tor Vergata} in Rome, and the {\em University of L'Aquila}.

\noindent Research partially supported by COFIN, Prin n.20078XYHYS Faperj grants E-26/100.626/2007 and E-26/170.008/2008.
I.M.'s work was partially supported by GNFM young researchers project
``Statistical mechanics of multicomponent systems". M.E.V. is partially supported by CNPq grant 302796/2002-9.
\medskip

\bibliographystyle{amsalpha}

\bigskip
\bigskip
\begin{appendices}
\noindent\section{Estimates for $\cG(m|\bar \si)$}
\label{app:appA0}

\noindent Recall \eqref{def:cG}, which we write as $\mathcal{G}(m\mid\bar \sigma)=\mathcal{G}^{(1)}_\La(m) + \cG^{(2)}_\La(m\mid\bar \si)$ with
\begin{eqnarray}
\label{def:cG-app}
&&\cG^{(2)}_\La(m\mid\bar \si):= \frac{\ga}{\beta}\ln \#\{\si_\La\in \cS_\La: m^{\ell_0}(\si_\La)=m \}
\\ \nn
&&  -\frac{\ga}{\beta}\ln\bigg(\sum_{\si: m^{\ell_0}(\si)=m}
\exp\bigg\{-\frac{\beta}{2}\sumtwo{x, y\in\ell_0\Zz\cap \La}{i\in C_x^0,j\in C_y^0}
(J_\ga(|i-j|)-J_\ga(|x-y|))\si_i\si_j \nn \\
&& -\beta\sumtwo{x \in \ell_0 \mathbb{Z}\cap \La, y \ell_0\mathbb{Z}\setminus \La}{i\in C_x^0,j\in C_y^0}(J_\ga(|i-j|)-J_\ga(|x-y|))\si_i\bar\si_j
 \bigg\} \bigg). \nn
\end{eqnarray}

\begin{lemma}
   \label{lemma:G1}
Let $\Delta\subset \La$ be an $\ell_+$-measurable interval. Assume that $m\in\cE_\Delta(\und 1)$ and $\bar \si\in \cS^+$.
Taking
\begin{eqnarray*}
m^*(x)=\left\{
        \begin{array}{ll}
          m(x), & \hbox{if } x\notin \Delta\cap \ell_0\Zz\\
          -m(x), & \hbox{if } x\in \Delta\cap \ell_0\Zz,\\
        \end{array}
      \right.
\end{eqnarray*}
we have
    \begin{eqnarray}
\label{eq:G1}
|\cG_\La(m|\bar \si)-\cG_\La(m^*|\bar \si)|<4
\delta_0 [\tilde \la K +1]
\end{eqnarray}
for some absolute constant $K$.
\end{lemma}

\begin{lemma}
   \label{lemma:G2}
Let $\Delta\subset \La$, $m\in \cE_\Delta(0)$, and
\begin{eqnarray*}
m^*(x)=\left\{
        \begin{array}{ll}
          m(x), & \hbox{if } x\notin \Delta\cap \ell_0\Zz\\
          \tilde m(x), & \hbox{if } x\in \Delta\cap \ell_0\Zz 
        \end{array}
      \right.
\end{eqnarray*}
where $\tilde m\in \cM_{0,\Delta}$ is an arbitrary profile.
Then
    \begin{eqnarray}
\label{eq:G2}
|\cG_\La(m|\bar \si)-\cG_\La(m^*|\bar \si)|<
2\delta_0(1+\ga \la K )+ \ga(5\la +2\delta_0)|\Delta|
\end{eqnarray}
for $K$ as in \eqref{eq:G1}.
\end{lemma}

\begin{proof}[Proof of Lemma \ref{lemma:G1}]
Notice that $\cG^{(1)}_\La(m|\bar \si)=\cG^{(1)}_\La(m^*|\bar\si)$. Also, due to the spin flip symmetry, the contribution to  $\cG^{(2)}_\La(m|\bar \si)-\cG^{(2)}_\La(m^*|\bar\si)$ corresponding to $x, y \in \La \cap \ell_0\mathbb{Z}$ vanishes. On the other hand, for $|| i- j|- | x - y||<  \ell_0 $ we get:
\begin{eqnarray*}
|J_\ga(i,j)- J_\ga(x,y)|<
\max\{\ga, 5\la\ga^2\} \Ii_{[||x-y|-1/(2\ga)|<2\ell_0]}+\frac{3\lambda\ell_0}{|x-y|^{3}}\Ii_{[|x-y|>1/(2\ga)+2\ell_0]}
\end{eqnarray*}
for any $i\in C^0_x, j\in C^0_y$, $x,y\in \ell_0\Zz $. Simple computations show that
\begin{eqnarray*}
|\cG^{(2)}_\La(m|\bar \si)-\cG^{(2)}_\La(m^*|\bar\si)|<\frac{\ga}{\beta} \ln \left(\frac{\sum_{\si: m^{\ell_0}(\si)=m} e^{\beta[2\ell_0+\la K\delta_0]}}{\#\{\si\in \cS_\La: m^{\ell_0}(\si)=m \}}\right)<\delta_0[2+\la K\ga]
\end{eqnarray*}
where $K$ is an absolute  constant.
\end{proof}

\begin{proof}[Proof of Lemma \ref{lemma:G2}]
Let
\begin{equation*}
\cH_A(\si):=\frac{1}{2} \sum_{x, y\in A\cap\ell_0 \Zz}\sumtwo{i\in C_x^0}{j\in C_y^0}
(J_\ga(|i-j|)-J_\ga(|x-y|))\si_i\si_j
\end{equation*}
and
\begin{equation*}
\cH_{A,B}(\si):=\sumtwo{x\in A\cap\ell_0 \Zz}{y\in B\cap\ell_0 \Zz}
\sumtwo{i\in C_x^0}{j\in C_y^0}(J_\ga(|i-j|)-J_\ga(|x-y|))\si_i\si_j.
\end{equation*}
We may write\footnote{$\si\circ \bar \si$ denotes the configuration that agrees with $\si$ in $\La$ and with $\bar \si$ outside
$\Lambda$.}:
\begin{eqnarray*}
&&|\cG^{(2)}_\La(m|\bar \si)-\cG^{(2)}_\La(m^*|\bar \si)|\\
&&\hskip1cm\le\bigg|\frac{\ga}{\beta}\ln\left(2\sup_{\si}e^{-\beta \cH_{\Delta ,\Delta^c}(\si\circ \bar \si)}\frac{
\sum_{\si: m^{\ell_0}(\si)=m}e^{-\beta (\cH_{\Delta}(\si)+ \cH_{\Delta^c}(\si\circ \bar \si))}}{\sum_{\si: m^{\ell_0}(\si)=m^*}
\e^{-\beta (\cH_{\Delta}(\si)+\cH_{\Delta^c}(\si\circ \bar \si))}}\right)\\
&& \hskip 3cm
+\frac{\ga}{\beta}\ln \left(\frac{\sum_{\si_{\Delta}: m^{\ell_0}(\si_{\Delta})=m^*_{\Delta}}
 1}{\sum_{\si_{\si_{\Delta}: m^{\ell_0}(\si_{\Delta})=m_{\Delta}} 1}}\right)\bigg|
\\
&&\hskip1cm\le \bigg|\frac{\ga}{\beta}\ln\left(2\sup_{\si}e^{-\beta \cH_{\Delta ,\Delta^c}(\si\circ \bar \si)}
\sup_{\si_{\Delta}: m^{\ell_0}(\si_{\Delta})=m_{\Delta}}e^{-\beta \cH_{\Delta}(\si)} \sup_{\si_{\Delta}: m^{\ell_0}(\si_{\Delta})=m^*_{\Delta}} e^{+\beta \cH_{\Delta}(\si)}\right)
 \bigg|
\\
&&\hskip1cm\le2\delta_0\bigg|(1+\ga \la K )
+\ga|\Delta| \bigg|.
\end{eqnarray*}
Trivial computations and Stirling formula give
\begin{eqnarray*}
|\cG^{(1)}_\La(m|\bar \si)-\cG^{(1)}_\La(m^*|\bar \si)|<5\la\ga|\Delta|.
\end{eqnarray*}
\end{proof}

\vskip 0.5cm \noindent
\section{Proof of Proposition \ref{prop:mu-rectangles}}
\label{app:}

\begin{proof}[Proof of the proposition \ref{prop:mu-rectangles}]

Recalling the definitions \eqref{def:mua}, \eqref{def:mub}, we
will prove that:
\begin{eqnarray}
\epsilon_a> c^\prime_\beta \psi^3\delta_-
\label{mua-bound}
\\
\epsilon_b>  \frac14 m_\beta^2\delta_-
\label{mub-bound}
\end{eqnarray}
where $c^\prime_\beta$ is a positive constant. Once this is proven, and since  $\psi$ is chosen as indicated right after (\ref{def:eta}), we get \eqref{eq-prop:mu-rectangles}.

\vskip 0.5cm
\noindent We first prove \eqref{mub-bound}. Hence we are assuming that $\eta(h)=0$ and:
\begin{eqnarray*}
||m^{\ell_-}(z)|-m_\beta|<\psi \;\;\; \forall z\in C^+_h\cap\ell_-\Zz.
\end{eqnarray*}
Recalling \eqref{def:Fm} and writing $m(x)$ for $m^{\ell_0}(x)$, we have:
  \begin{eqnarray*}
\label{eq: fc1}
F_{C^+_h}(m|\bar m)-F_{C^+_h}(m_\beta|m_\beta)&\ge&  \frac{\delta _0^2}{2} \sum_{x,y\in C^+_h\cap\ell_0\Zz}(m(x)-m(y))^2 \Ii_{[|\delta_0x-\delta_0 y|<1/2]}\\
&\ge& \sumtwo{u,v\in\ell_0\Zz\cap C^+_h}{|u-v|<\ga^{-1}/3} I_{u,v}(m)
\end{eqnarray*}
where
\begin{eqnarray}
\label{eq; I}
I_{u,v}(m):=\frac{\delta_0^2}{2}\sum_{x\in C^-_u} \sum_{y\in C^-_v}(m(x)- m(y))^2.
\end{eqnarray}

For each $u,v$ such that $|m^{\ell_-}(u)-m_\beta|<\psi$ and $|m^{\ell_-}(v)+m_\beta|<\psi$:
\begin{eqnarray*}
I_{uv}(m)&\ge&\frac{\delta_0^2}{2}\sum_{x\in C^-_u} \sum_{y\in C^-_v}(m(x)- 2m_\beta +2m_\beta-m(y))^2
\\
&\ge&\frac{\delta_0^2}{2}\sum_{x\in C^-_h} \sum_{y\in C^-_k}[(m(x)- m_\beta )^2 +(2m_\beta)^2+ (m(y)+m_\beta)^2- 6\psi]
\\
&\ge&\frac{\delta_-^2}{2}[(2m_\beta)^2- 6\psi].
\end{eqnarray*}

Since $\eta(h)=0$ there are at least  $\frac{1}{3\delta_-}-1$  such pairs $(u,v)$,
and we have the following lower bound:

 \begin{eqnarray*}
F_{C^+_h}(m|\bar m)-F_{C^+_h}(m_\beta|m_\beta)
\ge  \frac{\delta_-}{4}[(2m_\beta)^2- 6\psi]\ge  \frac{\delta_-}{4}m_\beta^2
\end{eqnarray*}
for $\psi<m_\beta^2/2$.

\vskip 0.5cm \noindent

Let us now consider $\epsilon_a$.  In this case, we have at least a block $C^-_{z}$ where \linebreak
$||m^{\ell_-}(z)|-m_\beta|>\psi$. The  main contribution to the free energy in this case comes
from the local contribution on the blocks $C^{0}_x$ with $x\in \ell_0\Zz\cap C^-_{z}$ where $|m(x)-m_\beta|>\psi$, and/or from the interaction
between two blocks $C^{0}_x, C^{0}_y$, $x,y\in \ell_0\Zz\cap C^-_{z}$, with magnetization of opposite signs, close to $\pm m_\beta$.
Recalling that $\delta_-<1/2$,  this last term is due only to the short range interaction, which is constant inside  each block $C^-$.
We consider a lower bound of the free energy by neglecting all other (non-negative) contributions. For this we set:
\begin{eqnarray*}
\cN_0&:=&\{x\in C^-_z \cap\ell_0\mathbb{Z}\colon ||m(x)|-m_\beta|>\psi/2\}
\\
\cN_\pm&:=&\{x\in  C^-_z \cap\ell_0\mathbb{Z}\colon |m(x)-\pm m_\beta|<\psi/2\}.
\end{eqnarray*}

Let $N_0, N_\pm$  denote the cardinality of the sets $\cN_0, \cN_\pm$ respectively, and $n_\pm:=N_\pm\frac{\ell_0}{\ell_-}$.
Hence $N_0+N_+ +N_-=\frac{\ell_-}{\ell_0}$. It is trivial to verify that due to our condition on
$m^{\ell_-}(z)$  we have
\begin{eqnarray}
\label{proptA}
n_\pm\le \left(1-\frac{\psi}{4}\right),
\end{eqnarray}
and we can write (with $n_0=1-n_--n_+$)
\begin{eqnarray*}
F_{C^-_z}(m|\bar m)-F_{C^-_z}(m_\beta|m_\beta)&\ge& n_0 \delta_- (f( m_\beta+\psi) - f(m_\beta))+\delta_-^2 n_-n_+  (2m_\beta-\psi)^2 \nn
\\
&\ge&  \delta_- (1-n_--n_+)c_\beta \frac{\psi^2}{4}+
\delta_-^2 n_-n_+ (2m_\beta-\psi)^2
\end{eqnarray*}
where $c_\beta$ is a lower bound of the second derivative of the mean field free energy $f$.

We then can take the minimum of the r.h.s of the above equation, on the set $\{(n_-,n_+) \colon 0\le n_\pm\le \left(1-\frac{\psi}{4}\right), n_-+n_+\le 1\}$, which gives:
\begin{eqnarray*}
F_{C^-_z}(m|\bar m) - F_{C^-_z}(m_\beta|m_\beta)
&\ge&  \delta_-c'_\beta \psi^3
\end{eqnarray*}
for a suitable positive constant $c'_\beta$.
In fact, the function $g(x,y)=A(1-x-y) +Bxy$ with  $A=\delta_-c_\beta \psi^2$ and $B=\delta_-^2(2m_\beta-\psi)^2$
has a unique critical point in $x=y=\frac{A}{B}$, that, since $B\ll A$, is out of the  domain
 $X+Y\le 1$. This point is  a {\em saddle}. Evaluating the function on the border of the domain, we get the result:
$g(x,y)\ge g(0,(1-\psi/4))= g((1-\psi/4),0)= A\frac{\psi}{4}= \delta_-c'_\beta \psi^3$.
\end{proof}

\vskip 0.5cm \noindent
\section{Proof of equation \eqref{eq:errico}}
\label{app:mu versus Wtildem}

\begin{lemma}
\label{thm:errico}
Let $\Delta\subset \La$ be an ${\ell_0}$-measurable interval. Then there is a constant $c$ and
for any $\bar m_{\Delta^c}\in \mathcal{M}_{0,\Delta^c}$ there exists an $\ell_0$-measurable function $\phi_\Delta(x)\equiv
\phi_\Delta(x;m_{\Delta^c})$ so that:

\begin{equation}
\label{eq:errico0}
\inf_{m_\Delta\in \cE_\Delta(\eta_\Delta=+1)}F_\Delta(m_\Delta|\bar m_{\Delta^c})\ge
F_\Delta(\phi_\Delta(x)|\bar m_{\Delta^c})-c\la\gamma \ln|\Delta|
\end{equation}
where in \eqref{eq:errico0} we  $F_\Delta(m_\Delta|\bar m_{\Delta^c})$ has been extended to $(-1,1)$--valued profiles in the
obvious way.

\begin{equation}
\label{eq:errico1}
|\phi_\Delta(x;m_{\Delta^c})-m_\beta|< C \; e^{-c(\beta)\gamma d(x,\Delta^c)}
\end{equation}
uniformly in $\bar m_{\Delta^c}$.
\end{lemma}

\begin{proof}[Proof of Lemma \ref{thm:errico}]
After remarking (a) and (b) below, the proof is essentially that in \cite{e-book} except for the fact that our
$J^{(0)}$ is not smooth. The technical details to adapt the proof to this
case have already been taken care in \cite{BKMP2}, see Appendix D there.

\noindent (a) For any $\bar m_{\Delta^c}(x)$, and if $F^0_\Delta(m_\Delta|\bar m_{\Delta^c})$ denotes
the functional $F_\La$ calculated for $\la=0$, one has
\begin{equation*}
F_\Delta(m_\Delta|\bar m_{\Delta^c})\ge F^0_\Delta(m_\Delta|\bar m_{\Delta^c})
\end{equation*}

\noindent (b) $|F_\Delta(m_\Delta|\bar m_{\Delta^c})-F^0_\Delta(m_\Delta|\bar m_{\Delta^c})|<c\la\gamma \ln|\Delta|$.
\end{proof}

\noindent We now state the following corollary of Lemma \ref{thm:errico}.
\begin{corol}
   \label{corol:errico}
Let us assume that  $\eta(h)=+1$ and let $\Delta=C^+_h$. Then the infimum of \eqref{eq:errico0}
is achieved on a function $\phi_\Delta$ that satisfies

\begin{equation}
\label{eq:corol-errico}
|\phi_\Delta(x)-m_\beta|<e^{-c\delta_+/3}=: \eps(\ga)\equiv\eps
\hskip 1cm \text{if } \dist(x,\Delta^c)>\ell_+/3.
\end{equation}\end{corol}

\vskip 0.5cm \noindent
\section{Proof of (\ref{eq:triangulo})}
\label{app:triangulo}

\noindent Recall that $T$ is the smallest triangle in the configuration $\und \Ga(\sigma)=\{\Ga\}$.

\begin{proof} It follows from equation \eqref{def:Fauxiliar} and the symmetry of $f_\beta(m)$ that:
\begin{eqnarray*}
F_{ T,\Delta^c}(m\circ\bar m)-F_{ T,\Delta^c}(m^*\circ\bar m)
&=&-\gamma(\ell_0)^2\sumtwo{x\in \ell_0\Zz\cap { T}}{y\in \ell_0\Zz\setminus {\Delta}}
J_\gamma(|x-y|)m(x) m(y),
\end{eqnarray*}
and by Lemma \ref{thm:interazioneAB} below we can write:

\begin{eqnarray*}
F_{ T,\Delta^c}(m\circ\bar m)-F_{ T,\Delta^c}(m^*\circ\bar m)\ge
-\gamma (\ell_-)^2\sumtwo{u\in \ell_-\Zz\cap { T}}{v\in \ell_-\Zz\setminus {\Delta}}
J_\gamma(|u-v|)m^{\ell_-}(u) m^{\ell_-}(v)-8\tilde\la K\delta_-.
\end{eqnarray*}

Under the hypothesis of Lemma \ref{thm:k},
$|m^{\ell_-}(u) +s\cdot m_\beta| <\psi$ for $u\in  T$, where $s=\sign(T)$. On the other hand,
\begin{eqnarray*}
  m^{\ell_-}(v)
  \begin{cases}
    \in (s\cdot m_\beta-\psi,  s\cdot m_\beta+\psi) & \text{if }  v\in I(T) \\
    \in \cM_- & \text{if }  v\notin I(T)
  \end{cases}
\end{eqnarray*}
from which \eqref{eq:triangulo} follows easily.
\end{proof}

\begin{lemma}
   \label{thm:interazioneAB}
There exist a constant $K$ so that for every $A$, $B$ disjoint $\ell_+$--measurable intervals on $\Zz$ and every
$m(\cdot) \in \mathcal{M}_{0,\mathbb{Z}}$:
\begin{eqnarray}
\label{eq:interazioneAB}
&&\hskip-1cm\nn
\bigg|(\ell_0)^2\sumtwo{x\in A\cap \ell_{0}\Zz}{y\in B\cap\ell_{0}\Zz}{J}_\gamma(|x-y|)m(x)m(y)
-(\ell_-)^2\sumtwo{u\in A\cap \ell_{-}\Zz}{v\in B\cap \ell_{-}\Zz} J_\gamma(|u-v|)
m^{\ell_{-}}(u)m^{\ell_{-}}(v)\bigg|
\\
&&\hskip3cm
\le\ell_-[4\delta_-\Ii_{[d(A,B)=1]}+8\la K]
\end{eqnarray}
\end{lemma}

\begin{proof}
The l.h.s. of (\ref{eq:interazioneAB}) can be written as
\begin{equation*}
\big|(\ell_0)^2\sumtwo{u\in \ell_-\Zz\cap {A}}{v\in \ell_-\Zz\cap {B}}
\sumtwo{x\in \ell_0\Zz\cap{C^-_u}}{y\in \ell_0\Zz\cap{C^-_v}}
\Delta J_\gamma(x,y)m(x)m(y)\big|
\end{equation*}
where for $x \in C^-_u, y\in C^-_v$, $J_\gamma(|x-y|)-J_\gamma(|u-v|)$. Direct calculation shows that:
\begin{eqnarray*}
|\Delta J_\gamma(x,y)| \le \gamma \Ii_{[||u-v|-1/(2\gamma)|\le 2\ell_-]}+\frac{8\la \ell_-\Ii_{[|u-v|>1/(2\gamma)-2\ell_-]}}{|u-v|^{3}}
\end{eqnarray*}
from where (\ref{eq:interazioneAB}) follows.
\end{proof}
\end{appendices}
\end{document}